\documentclass[final]{dmtcs-episciences}

\usepackage[english]{babel} 
\usepackage[utf8]{inputenc}
\usepackage{subfigure}
\usepackage[round]{natbib}

\author{Cyril Banderier\affiliationmark{1}\thanks{\url{http://lipn.fr/~banderier}} 
  \and Jean-Luc Baril\affiliationmark{2}\thanks{\url{http://jl.baril.u-bourgogne.fr/}}
  \and C\'eline Moreira Dos Santos\affiliationmark{2}\thanks{\url{http://le2i.cnrs.fr/-Celine-Moreira-Dos-Santos-}}}

\title[Right-jumps and pattern avoiding permutations]{Right-jumps \& pattern avoiding permutations}
\affiliation{
  LIPN UMR-CNRS 7030, Universit\'e de Paris Nord, France.\\
LE2I UMR-CNRS 6306, Universit\'e de Bourgogne,  France.
}

\received{2015-12-8}
\accepted{2017-2-8}
\revised{2017-2-8}
\publicationdetails{18}{2017}{2}{12}{1344}

\usepackage{comment}
\usepackage{amsmath,amsthm}
\usepackage{xcolor}
\usepackage{hyperref}
\definecolor{darkgreen}{rgb}{0,0.4,0}
\definecolor{darkblue}{rgb}{0.0, 0.0, 0.55}
\definecolor{BrickRed}{rgb}{0.65,0.08,0}
\hypersetup{colorlinks=true,linkcolor=blue,citecolor=darkblue,filecolor=BrickRed,urlcolor=darkgreen}
\usepackage{tikz}
\usetikzlibrary{decorations.pathreplacing}

\newtheorem{Definition}{Definition}

\newtheorem{lem}{Lemma}
\newtheorem{cor}{Corollary}
\newtheorem{thm}{Theorem}

\def \Avoid{\operatorname{Avoid}}
\def \B{{\mathcal{B}}}
\def \S{{\mathcal{S}}}
\renewcommand{\mod}{\operatorname{mod}}
\usepackage{stmaryrd}  \newcommand{\interval}[1]{\llbracket #1 \rrbracket} 

\newcommand{\oeis}[1]{\text{\href{https://oeis.org/#1}{{\small \tt OEIS #1}}}} 
\newcommand{\OEIS}[1]{\text{\href{https://oeis.org/#1}{{\small \tt (OEIS #1)}}}} 

\begin{document}
\maketitle

\begin{abstract}
We study the iteration of the process of moving values to the right in permutations.
We prove that the set of permutations obtained in this model after a given number of iterations from the identity is a class of pattern avoiding permutations.
We characterize the elements of the basis of this class and enumerate it by giving their bivariate exponential generating function:
we achieve this via a catalytic variable, the number of left-to-right maxima. We show that this generating function is a D-finite function satisfying a
differential equation of order~2. We give some congruence properties for the coefficients of this generating function,
and  show that their asymptotics involves a rather unusual algebraic exponent (the golden ratio $(1+\sqrt 5)/2$) and some closed-form constants.
We end by proving a limit law: a forbidden pattern of length $n$ has typically $(\ln n) /\sqrt{5}$ left-to-right maxima, with Gaussian fluctuations.
\end{abstract}


\section{Introduction}

In computer science, many algorithms related to sorting a permutation have been analysed
and shown to have behaviours linked to nice combinatorial properties~(see e.g.~\cite{Knuth}).
Their complexity can be analysed in terms of memory needed, or number of key operations
(like comparisons or pointer swaps). An important family of algorithms, like the so-called insertion algorithms, or in situ permutations,
are quite efficient in terms of the number of pointer swaps (but are not the fastest ones in terms of comparisons).
Due to this higher cost, they have been much less studied than the faster stack sorting algorithms.
Like for the stack algorithms, instead of seeing them as an input/output pair, we can see them like a {\em process}: input and set of intermediate steps.
This opens a full realm of questions on such processes, and they often lead to nice links with other parts of mathematics
(like the link between trees, birth and death processes, random walks in probability theory, or permutations and Young tableaux in algebraic combinatorics).
Our article will investigate a link between a sorting algorithm, patterns in permutations, and their asymptotics counterparts.

\bigskip

Another motivation to analyse such processes comes from the field of bioinformatics. Indeed, in genomics, a crucial part of study is to estimate the similarity of two genomes.
This consists in finding the length of a shortest path of evolutionary mutations that transforms one genome into another.
Usually, the main operations used in the rearrangement of a genome are of three different types: substitutions (one gene is replaced with another), insertions (a gene is added)
and deletions (a gene is removed). For instance, we refer to~\cite{Jon,Mei} for an explanation of these operations, and the notions of transposons or jumping genes.

As the problem is too hard in full generality, many simpler mathematical models of the genome are used  (see~\cite{Wat}): one of them is using
permutations of $\{1,2,\ldots, n\}$ where each gene is assigned a number.
Following the idea of transposition mutations (see~\cite{Jon}), our motivation is to find some combinatorial properties in terms of pattern
avoiding permutations whenever one element is deleted and inserted in a position to its right.
This operation will be called a right-jump.

\begin{figure}[h]
		\begin{center}
			\begin{tikzpicture}[scale=0.4]
				\draw (1.5cm,0cm) node { 1};
				\draw (2.5cm,0cm) node { 2};
				\draw[color=green] (3.5cm,0cm) node { 3};
				\draw (4.5cm,0cm) node {4};
				\draw (5.5cm,0cm) node {5};
				\draw (6.5cm,0cm) node { 6};
				\draw (7.5cm,0cm) node { 7};
				\draw[dotted] (3cm,-0.5cm) -- (4cm,-0.5cm) -- (4cm,0.5cm) -- (3cm,0.5cm) -- (3cm,-0.5cm);
				\draw[dotted,->] (3.5cm,0.5cm) -- (3.5cm,1cm) -- (7cm,1cm) -- (7cm,0.5cm);
				\draw (9.5cm,0cm) node{{\Large $\rightsquigarrow$}};
				\draw (11.5cm,0cm) node {1};
				\draw (12.5cm,0cm) node {2};
				\draw (13.5cm,0cm) node { 4};
				\draw (14.5cm,0cm) node { 5};
				\draw (15.5cm,0cm) node { 6};
				\draw[color=green] (16.5cm,0cm) node { 3};
				\draw (17.5cm,0cm) node { 7};
 \end{tikzpicture}
		\end{center}
	\caption{\label{fig1}A right-jump in the permutation $\sigma=1234567$.
In this article, we investigate the structure of permutations obtained after several iterations of such right-jumps.
}
	 \end{figure}
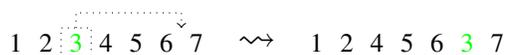
\bigskip

This operation is a variant of genome duplication, which consists of copying a part of the original genome inserted into itself, followed by the loss of one copy of each of the duplicated genes. In particular, it is comparable to the whole duplication-random loss model studied in~\cite{Chau}.
Although there are many connections between these models, it is surprising that the behaviour of their combinatorial properties depends on different parameters:
 the right-jump model reveals some links with {\it left-to-right maximum} statistics~(see \cite{JLB2,Bon}), while the whole duplication-random loss model reveals links with {\it descent} statistics (see~\cite{BV,BF,BP,BR,Chau,MY}).

 In the literature, such right-jumps in permutations are also found in the domain of sorting theory.
 Indeed, it corresponds (modulo a mirror symmetry) to the {\it insertion-sorting algorithms} on permutations (see~\cite{Knuth}).
Since the seminal work of Knuth on this subject, many articles related to sorting with a stack exhibit links with pattern avoiding permutations.
In contrast, for insertion sorting algorithms (also the subject of a vivid literature), only one study exhibits links with pattern avoiding permutations. 
Indeed, in his thesis, \cite{Mag} proves that the set of permutations that can be sorted with one step of the insertion-sorting operator
is the class of permutations avoiding the three patterns $321$, $312$ and $2143$.

\smallskip

{\bf Plan of the article.}
In Section 2, we recall some basic facts on permutations patterns.
In Section 3, we generalize the result of Magn{\'u}sson by studying the iteration of right-jumps in terms of pattern avoiding permutations: we prove that the set of permutations obtained from the identity after a given number of right-jumps is the class of permutations avoiding some patterns, which we characterize.
In Section 4, we enumerate these forbidden patterns by giving their bivariate exponential generating function (involving an additional parameter: the number of left-to-right maxima), and we give the corresponding asymptotics and limit law.
We also give some modular congruences for our main enumeration sequence.
In Section~5, we conclude with several possible extensions of this work.

\section{Patterns in permutations}\label{sec2}

In this section, we give some classical definitions and properties on patterns in permutations.
For any permutation $\sigma \in \S_n$ (the set of permutations of length $n$),
the {\it graphical representation} of $\sigma=\sigma_1\sigma_2\dots\sigma_n$  is the set of points in the plane at coordinates $(i,\sigma_i)$ for $i\in \interval{n}$\footnote{In this article, we write $\interval{n}$ for $\{1,2,\ldots, n\}$.}. For instance, the permutation $53621487$ has the graphical representation illustrated in Figure~\ref{fig2}.
A {\it left-to-right maximum} of $\sigma\in \S_n$ is a value $\sigma_i$, $1\leq i\leq n$, such that $\sigma_j\leq \sigma_i$ for $j\leq i$. A value $\sigma_i$ of $\sigma$, $1\leq i\leq n$ which is not a left-to-right maximum will be called a {\it non-left-to-right-maximum} of $\sigma$. For instance, if $\sigma=53621487$ then the left-to-right maxima are $5,6,8$ and the non-left-to-right-maxima are $1,2,3,4,7$. 

\smallskip

\def\a{0cm} \def\A{0.5cm}
\def\b{1cm}
\def\c{2cm}
\def\d{3cm}
\def\e{4cm}
\def\f{5cm}  \def\F{5.5cm}
\def\g{6cm} \def\G{6.5cm}
\def\h{7cm} \def\H{7.5cm}
\def\i{8cm} \def\I{8.5cm}
\def\j{9cm}
\def\k{10cm}
\def\l{11cm}
\def\m{12cm}
\def\n{13cm}
\def\o{14cm}
\def\p{15cm}
\def\q{16cm}
\def\r{17cm}
\begin{figure}[ht]
    \label{figPermutation}
        \begin{center}
            \begin{tikzpicture}[scale=0.8]
                \draw[thin, lightgray] (\b,\b) grid (\j,\j);
                \draw[below] (1.5cm,\b) node {\tiny 1};
                \draw[below] (2.5cm,\b) node {\tiny 2};
                \draw[below] (3.5cm,\b) node {\tiny 3};
                \draw[below] (4.5cm,\b) node {\tiny 4};
                \draw[below] (\F,\b) node {\tiny 5};
                \draw[below] (\G,\b) node {\tiny 6};
                \draw[below] (\H,\b) node {\tiny 7};
                \draw[below] (\I,\b) node {\tiny 8};

                \draw[left] (\b,1.5cm) node {\tiny 1};
                \draw[left] (\b,2.5cm) node {\tiny 2};
                \draw[left] (\b,3.5cm) node {\tiny 3};
                \draw[left] (\b,4.5cm) node {\tiny 4};
                \draw[left] (\b,\F) node {\tiny 5};
                \draw[left] (\b,\G) node {\tiny 6};
                \draw[left] (\b,\H) node {\tiny 7};
                \draw[left] (\b,\I) node {\tiny 8};

                \draw(1.5cm,\F)--(2.5cm,3.5cm);
                \draw(2.5cm,3.5cm)--(3.5cm,\G);
                \draw(3.5cm,\G)--(4.5cm,2.5cm);
                \draw(4.5cm,2.5cm)--(\F,1.5cm);
                \draw(\F,1.5cm)--(\G,4.5cm);
                \draw(\G,4.5cm)--(\H,\I);
                \draw(\H,\I)--(\I,\H);

                \fill[green](1.5cm,\F) circle (2mm);
                \fill[green](3.5cm,\G) circle (2mm);
                \fill[green](\H,\I) circle (2mm);

                 \fill[black] (2.5cm,3.5cm) circle (1mm);

                \fill[black](4.5cm,2.5cm) circle (1mm);
                \fill[black](\F,1.5cm) circle (1mm);
                \fill[black](\G,4.5cm) circle (1mm);
                \fill[black](\I,\H) circle (1mm);

                \draw[blue,dashed] (2.5cm,3.5cm) -- (\F,1.5cm) -- (\I,\H);
 \end{tikzpicture}
        \end{center}
    \caption{The graphical representation of $\sigma=53621487$. We show an occurrence of a pattern $213$ with a dashed  line;
    big green points are the left-to-right maxima and small black points are non-left-to-right-maxima.}
     \label{fig2}
\end{figure}
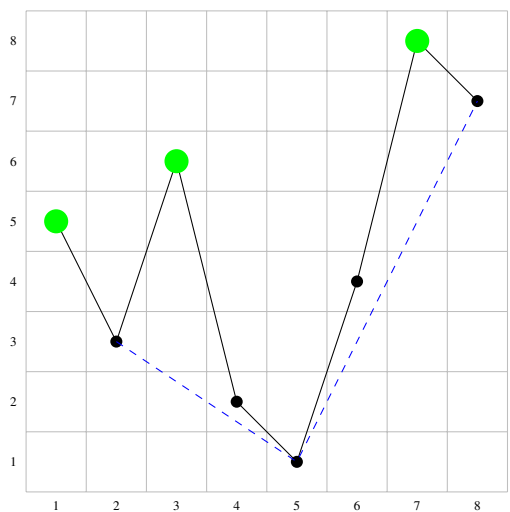

\smallskip
A permutation $\pi$ of length $k$,
is a {\it pattern} of a permutation $\sigma\in \S_n$ if there is a subsequence of $\sigma$ which is order-isomorphic to $\pi$, {\it i.e.}, if there is a subsequence $\sigma_{i_1}\ldots \sigma_{i_k}$ of $\sigma$ with $1\leq i_1< \dots <i_k\leq n$ and such that $\sigma_{i_j}<\sigma_{i_\ell}$ whenever $\pi_{j}<\pi_\ell$. We write $\pi\prec\sigma$ to denote that $\pi$ is a pattern of $\sigma$. A permutation $\sigma$ that does not contain $\pi$ as a pattern is said to {\it avoid} $\pi$. For example, $\sigma=2413$ contains the patterns $231$, $132$, $213$ and $312$, but $\sigma$ avoids the patterns $123$ and $321$. The set of all permutations avoiding the patterns $\pi_1,\ldots, \pi_m$ is denoted by $\Avoid(\pi_1,  \ldots, \pi_m)$.
We say that $\Avoid(\pi_1, \ldots, \pi_m)$ is a {\it class} of pattern avoiding permutations
with {\it basis} $\{\pi_1,\ldots, \pi_m\}$. For instance, we refer to the book of~\cite{Kit} and~\cite{Bon} to deepen these notions.
A set $\mathcal{C}$ of permutations is {\it stable} for the involvement relation~$\prec$ if, for any $\sigma\in \mathcal{C}$, for any $\pi\prec\sigma$, then we also have $\pi\in \mathcal{C}$.

\smallskip 
Now, we formulate a definition that is crucial for the present study.

\begin{Definition}[Permutation basis and basis permutations] \label{definition1} If a set $\mathcal{C}$ of permutations is stable for the involvement relation $\prec$, then $\mathcal{C}$ is a class of pattern avoiding permutations: $\mathcal{C} = \Avoid(\B).$
The basis $\B$ of forbidden patterns is then given by $\B=\{\sigma\notin\mathcal{C},\forall\pi\prec\sigma\mbox{ with } \pi\neq \sigma, \pi\in \mathcal{C}\}\,.$
In other words, the basis $\B$ is the set of {\it minimal} permutations $\sigma$ that do not belong to $\mathcal{C}$, where {\it minimal} is intended in the sense of the pattern-involvement relation $\prec$ on permutations, that is: if $\pi\prec \sigma$ and $\pi\neq \sigma$ then $\pi\in \mathcal{C}$. Notice that $\B$ might be infinite.
We call {\em basis permutations} the permutations belonging to $\B$.
\end{Definition}

Equipped with these definitions, our mission consists now in giving a description of the permutations belonging to the basis $\B_p$
(the permutations which are the minimal forbidden patterns) for the set $\mathcal{C}_p$ of permutations at distance at most $p$ from the identity, i.e., permutations obtained from the identity after at most $p$ right-jumps.

\section{Iteration of right-jumps in permutations: a structural description of the forbidden patterns}

 In this section we study the iteration of right-jumps in terms of pattern avoiding permutations.
We establish that the set $\mathcal{C}_p$ of permutations obtained from the identity after at most $p$ right-jumps is a class of permutations avoiding some patterns that we characterize.

\begin{lem}[Characterization of the distance]\label{lem1}
A permutation obtained from the identity after $p$ right-jumps contains at most $p$ non-left-to-right-maxima.
\end{lem}
\begin{proof} The result holds for $p=1$; indeed a right-jump transformation of the identity permutation creates the permutation $1~2~\ldots (i-1)~(i+1)\ldots (j-1)~i~j\ldots n$ for $1\leq i<j$, where $i$ is the only one non-left-to-right-maximum.
Now, let us assume that each permutation $\pi$ obtained from the identity after $(p-1)$ right-jumps contains at most $p-1$ non-left-to-right-maxima. Let $\sigma$ be a permutation obtained from the identity after $p$ right-jumps. Using the recurrence hypothesis, $\sigma$ is obtained from a permutation $\pi$ with at most $p-1$ non-left-to-right-maxima by moving an element $\pi_i$, $1\leq i< n$, in a position to its right.

We distinguish two cases: (1) $\pi_i$ is a non-left-to-right-maximum, and (2) $\pi_i$ is a left-to-right maximum.

Case (1): Since $\pi_i$ is a non-left-to-right-maximum, there exists $j<i$ such that $\pi_j$ is a left-to-right maximum satisfying $\pi_j>\pi_i$. Since we move $\pi_i$ to its right, $\pi_j$ remains on the left of $\pi_i$ in $\sigma$ which implies that $\pi_i$ is a non-left-to-right-maximum in $\sigma$. Using the same argument, any non-left-to-right-maximum $\pi_k$ in $\pi$ remains a non-left-to-right-maximum in $\sigma$. Moreover, let $\pi_k$ be a left-to-right maximum in $\pi$, {\it i.e.}, $\pi_j<\pi_k$ for all $j<k$. Since the
right-jump transformation moves to the right of a non-left-to-right-maximum, all values on the left of $\pi_k$ in $\sigma$ are lower than $\pi_k$, which proves that $\pi_k$ remains a left-to-right maximum in~$\sigma$. Therefore, $\sigma$ contains at most $p-1$ non-left-to-right-maxima (as $\pi$ does).

Case (2): $\pi_i$ is a left-to-right maximum, {\it i.e.}, $\pi_j<\pi_i$ for all $j<i$. Since $\pi_i$ is moved to its right, any left-to-right maximum located on the left of $\pi_i$ in $\pi$ remains a left-to-right maximum in $\sigma$. On the other hand, any left-to-right maximum located on the right of $\pi_i$ in $\pi$ is greater than $\pi_i$ and thus, it remains a left-to-right maximum in $\sigma$. Therefore, the number of left-to-right maxima in $\sigma$ is at least the number of left-to-right maxima in $\pi$ minus one (we do not consider $\pi_i$). This means that the number of non-left-to-right-maxima in $\sigma$ is at most $p$.

Considering the two previous cases allows to complete the proof, by induction.\end{proof}

We now derive our first enumeration result
 for the set $\mathcal{D}_p$ of permutations at distance $p$ from the identity, {\em i.e.}, the set of permutations reachable from the identity with $p$ right-jumps, but that one cannot reach with less than $p$ right-jumps:
 $\mathcal{C}_p = \cup_{k \in \interval{p}} \mathcal{D}_k$, and this union is disjoint.

\begin{thm}[Permutations after $p$ right-jumps] \label{thm1}
The set $\mathcal{D}_p$ of permutations at distance $p$ from the identity
is the set of permutations with exactly $p$ non-left-to-right-maxima.
Accordingly, the number $d_{n,p}$ of permutations of length $n$ in $\mathcal{D}_p$ is counted by the Stirling numbers $s(n,n-p)$:
$$d_{n,p}= s(n,n-p)= \sum_{0\leq j\leq h \leq p} (-1)^{j} \binom{h}{j} \binom{n-1+h}{p+h} \binom{n+p}{p-h}\frac{(j-h)^{p+h}}{h!}\,.$$
\end{thm}

\begin{proof}After considering Lemma~\ref{lem1}, it suffices to prove that any permutation $\sigma$ with at most $p$ non-left-to-right-maxima can be obtained from the identity after $p$ right-jumps. Let $\sigma$ be a permutation with $p\geq 1$ non-left-to-right-maxima. Let us assume that the leftmost non-left-to-right-maximum is $\sigma_i$ and let $j<i$ be the position of the smallest left-to-right maximum $\sigma_j$ such that $\sigma_j>\sigma_i$. Then we set $\sigma'=\sigma_1\ldots\sigma_{j-1}\sigma_i\sigma_j\ldots
\sigma_{i-1}\sigma_{i+1}\ldots\sigma_n$. Since we have $\sigma_j>\sigma_i$ and also $\sigma_i>\sigma_{j-1}$ (if $\sigma_{j-1}$ exists), $\sigma_i$ becomes a left-to-right maximum in $\sigma'$. Thus, $\sigma'$ contains exactly $p-1$ non-left-to-right-maxima and by construction, $\sigma$ can be obtained from $\sigma'$ by a right-jump.
This proves that permutations at distance $p$ from the identity are exactly the permutations with $n-p$ left-to-right-maxima,
which are known to be counted by $s(n,n-p)$, the signless Stirling number of the first kind (see~\cite{FlSe} for the closed-form formula due to Schl{\"o}milch,
and sequence \oeis{A094638} in~\cite{Slo} for many occurrences of  the corresponding triangular array).
\end{proof}

For instance, the values of $d_{n,p}$ for $n=7$ and $0\leq p <7$ are 1, 21, 175, 735, 1624, 1764, 720.

\medskip

The following corollary
says a little more on the lattice structure associated to our process "a particle jumps to the right".
\medskip

\begin{cor}[Changing the starting point and sorting algorithms]
\label{cor0} 
For any permutation $p$, one denotes by $t_p$ its number of non-left-to-right-maxima.
Let $\sigma$ and $\pi$ be two permutations, then $t_{\sigma^{-1}\cdot \pi}$ right-jumps are necessary to obtain $\pi$ from $\sigma$. In particular, $t_{\sigma^{-1}}$ right-jumps are necessary and sufficient to sort by insertion the permutation $\sigma$ into the identity.
\end{cor}

\begin{proof} Firstly, $t_\sigma$ right-jumps are necessary and sufficient to obtain $\sigma$ from the identity.
Therefore, $t_{\sigma^{-1}\cdot \pi}$ transformations are sufficient and necessary to obtain $\sigma^{-1}\cdot\pi$ from the identity. We set $t=t_{\sigma^{-1}\cdot \pi}$ and let $Id=\chi_0, \chi_1, \ldots, \chi_{t-1}, \chi_t=\sigma^{-1}\cdot\pi$ be a shortest path between the identity and $\sigma^{-1}\cdot\pi$.
 Now, let us prove that if a permutation $\beta$ is obtained from $\alpha$ by one right-jump, then for any permutation $\gamma$, $\gamma\cdot\beta$ is also obtained from $\gamma\cdot\alpha$ by one right-jump.
Indeed, if we have $\alpha=\alpha_1\alpha_2\ldots\alpha_n$ then $\beta$ can be written as $\beta=\alpha_1\ldots \alpha_{i-1}\alpha_{i+1} \ldots \alpha_{j-1}\alpha_i\alpha_j\ldots \alpha_n$. Composing by a permutation $\gamma$, we obtain $\gamma\cdot\alpha=\gamma(\alpha_1)\gamma(\alpha_2)\ldots\gamma(\alpha_n)$ and $\gamma\cdot\beta=\gamma(\alpha_1)\ldots \gamma(\alpha_{i-1})\gamma(\alpha_{i+1}) \ldots \gamma(\alpha_{j-1})\gamma(\alpha_i)\gamma(\alpha_j)\ldots \gamma(\alpha_n)$ which proves that $\gamma\cdot\beta$ is also obtained from $\gamma\cdot\alpha$ by one right-jump.
So if we compose by $\sigma$ at each step of the above shortest path, then we obtain a shortest path of $t_{\sigma^{-1}\cdot \pi}$ right-jumps from $\sigma$ to $\pi$, which completes the proof.\end{proof}
\medskip

Since the set $\mathcal{C}_p$ of permutations obtained after $p$ right-jumps is stable for the relation~$\prec$, $\mathcal{C}_p$ is also a class $\Avoid(\B_p)$ of pattern avoiding permutations where $\B_p$ is the basis consisting of minimal permutations~$\sigma$ that are not in $\mathcal{C}_p$ (see Definition~\ref{definition1}).
Theorem~\ref{thm2} gives the explicit description of these basis permutations.

\begin{thm}[Structural description of the basis permutations]\label{thm2}
A permutation $\sigma\in \S_n$ belongs to the basis ${\B}_p$ of forbidden patterns,
if and only if the following conditions hold:
\begin{enumerate}
\item[($i$)] $\sigma$ contains exactly $p+1$ non-left-to-right-maxima.
\item[($ii$)] $n-1$ is a non-left-to-right-maximum.
\item[($iii$)] $\sigma_2$ is a non-left-to-right-maximum.
\item[($iv$)] For any three left-to-right maxima, $\sigma_i$, $\sigma_j$ and $\sigma_k$ (with $i<j<k$) such that there is no left-to-right maximum between them, there exists a non-left-to-right-maximum $\sigma_t$ (with $j<t<k$) satisfying $\sigma_t>\sigma_i$.
\end{enumerate}
\end{thm}
\begin{proof} N.B.: Figure~\ref{fig3} on next page illustrates the different claims and notations of this theorem.

Let $\sigma\in \S_n$ be a permutation belonging to the basis $\B_p$, {\it i.e.}, $\sigma\notin\mathcal{C}_p$ and $\pi\prec \sigma$ implies $\pi\in \mathcal{C}_p$. Throughout this proof, we refer to Figure~\ref{fig3} for an illustration of the three conditions ($ii$), ($iii$) and ($iv$).

 - First, the deletion of a non-left-to-right-maximum in $\sigma$ decreases the number of non-left-to-right-maxima by one exactly. Therefore, the minimality of $\sigma$ implies that $\sigma$ necessarily contains exactly $p+1$ non-left-to-right-maxima, which proves ($i$).

 - For a contradiction, assume that ($ii$) is not satisfied, {\it i.e.}, $n-1$ is a left-to-right maximum. Since $n$ is always a left-to-right maximum, $n$ is on the right of $n-1$ in $\sigma$. Thus, the permutation $\pi$ obtained by deleting $n$ from $\sigma$ also contains $p+1$ non-left-to-right-maxima (a non-left-to-right-maximum on the right of $n$ in $\sigma$ remains a non-left-to-right-maximum on the right of $n-1$ in $\pi$). Therefore, $\pi$ does not belong to $\mathcal{C}_p$; this gives a contradiction with the minimality of $\sigma$.

 - For a contradiction, assume that ($iii$) is not satisfied, {\it i.e.}, $\sigma_2$ is a left-to-right maximum and thus, $\sigma_1$ is smaller than $\sigma_2$. Thus, the permutation $\pi$ obtained by deleting $\sigma_1$ from $\sigma$ also contains $p+1$ non-left-to-right-maxima. Indeed, a non-left-to-right-maximum $\sigma_i$ in $\sigma$ such that $\sigma_i<\sigma_1$ becomes a non-left-to-right-maximum $\sigma_i<\sigma_2-1$ in $\pi$. Moreover, a non-left-to-right-maximum $\sigma_i$ in $\sigma$ such that $\sigma_i>\sigma_1$ (there is $\ell$, $2\leq \ell<i$, with $\sigma_i<\sigma_\ell$) becomes a non-left-to-right-maximum $\sigma_i-1$ in $\pi$ with $\sigma_i-1<\sigma_\ell-1$. Therefore, $\pi$ does not belong to $\mathcal{C}_p$; this gives a contradiction with the minimality of $\sigma$.

 - For a contradiction, assume that ($iv$) is not satisfied; {\it i.e.}, there are $(i,j,k)$, $1\leq i<j<k\leq n$, such that $\sigma_i$, $\sigma_j$ and $\sigma_k$ are three consecutive left-to-right maxima of $\sigma$ (consecutive means that there is no other left-to-right maximum between $\sigma_i$ and $\sigma_j$ and between $\sigma_j$ and $\sigma_k$), and such that there is no non-left-to-right-maximum $\sigma_\ell$, $j<\ell<k$, satisfying $\sigma_i<\sigma_\ell$. Let $\pi$ be the permutation obtained from $\sigma$ by deleting $\sigma_j$.
 It is clear that any non-left-to-right-maximum on the left of $\sigma_j$ in $\sigma$ remains a non-left-to-right-maximum in $\pi$.
 Let $\sigma_\ell$, $\ell>k$, be a non-left-to-right-maximum on the right of $\sigma_k$ in $\sigma$. 
 If $\sigma_\ell<\sigma_j$ then $\sigma_\ell<\sigma_j\leq\sigma_k-1$ and $\sigma_\ell$ remains a non-left-to-right-maximum in $\pi$. If $\sigma_\ell>\sigma_j$ then there is $\sigma_{t}$, $t\geq k$, such that $\sigma_\ell\leq\sigma_t\geq\sigma_k$, and thus, there is $\sigma_t-1$ on the left of $\sigma_k-1$ in $\pi$ with $\sigma_t-1>\sigma_k-1$ which means that $\sigma_\ell-1$ is a non-left-to-right-maximum in $\pi$.
 Let $\sigma_\ell$, $j<\ell<k$, be a non-left-to-right-maximum between $\sigma_j$ and $\sigma_k$ in $\sigma$. Assuming that ($iv$) is not satisfied, we deduce that $\sigma_\ell<\sigma_i$, and $\sigma_\ell$ remains a non-left-to-right-maximum in $\pi$.
  Finally, $\pi$ also contains  $p+1$ non-left-to-right-maxima; this gives a contradiction with the minimality of $\sigma$.

Conversely, let $\sigma$ be a permutation satisfying ($i$), ($ii$), ($iii$) and ($iv$) and $\pi$ be a permutation obtained by deleting $\sigma_i$, $1\leq i\leq n$, from $\sigma$. Let us prove that $\pi$ belongs to $\mathcal{C}_p$, that is, $\pi$ contains at most $p$ non-left-to-right-maxima.

- If $\sigma_i$ is a non-left-to-right-maximum of $\sigma$, then $\pi$ has 
$p$ non-left-to-right-maxima and thus, $\pi\in\mathcal{C}_p$.

- Now, let us assume that $\sigma_i$ is a left-to-right maximum of $\sigma$. If $\sigma_i=n$, then ($ii$) implies that $n-1$ is a non-left-to-right-maximum of $\sigma$ and a left-to-right maximum in $\pi$; this implies that $\pi$ contains $p$ non-left-to-right-maxima and thus, $\pi\in\mathcal{C}_p$. If $\sigma_i=\sigma_1$, then ($iii$) implies that $\sigma_2$ is a non-left-to-right-maximum of $\sigma$ and $\sigma_2-1$ is a left-to-right maximum in $\pi$; this implies that $\pi$ contains $p$ non-left-to-right-maxima and thus, $\pi\in\mathcal{C}_p$.
If there exists $(j,k)$, $1\leq j<i<k\leq n$, such that $\sigma_j$ and $\sigma_k$ are left-to-right maxima (we choose $j$ the greatest possible and $k$ the lowest possible with this property). Then, ($iv$) implies that there is $\sigma_\ell$, $i<\ell<k$, such that $\sigma_j<\sigma_\ell<\sigma_i$ (we choose the lowest possible $\ell>i$). Thus, $\sigma_\ell$ is a non-left-to-right-maximum in $\sigma$ and becomes a left-to-right maximum in $\pi$, which implies that $\pi$ contains exactly $p$ non-left-to-right-maxima, and thus $\pi\in\mathcal{C}_p$.

Finally, the permutation $\pi$ necessarily belongs to $\mathcal{C}_p$, which completes the proof.\end{proof}

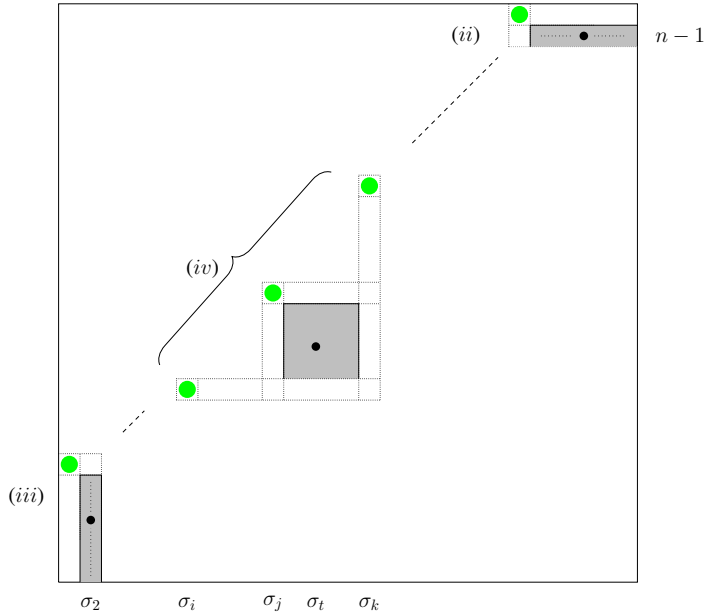
\begin{figure}[hbt]
    \begin{center}
            \scalebox{0.57}
            {\begin{tikzpicture}
                  \draw [very thin] (\a,-1cm) -- (\a,12.5cm) -- (13.5cm,12.5cm) -- (13.5cm,-1cm)--(\a,-1cm);
                \draw [fill=lightgray](\A,-1cm)--(\A,1.5cm)--(\b,1.5cm)--(\b,-1cm);
                \draw [fill=lightgray](5.25cm,3.75cm)--(5.25cm,\F)--(\h,\F)--(\h,3.75cm);
                \draw [fill=lightgray](11cm,11.5cm)--(11cm,12cm)--(13.5cm,12cm)--(13.5cm,11.5cm);
                \draw [dashed](1.5cm,2.5cm)--(\c,3cm);
                \draw [dashed](8.25cm,9.25cm)--(10.25cm,11.25cm);
                \draw [ densely dotted](\a,1.5cm)--(1cm,1.5cm);
                \draw [ densely dotted](\a,\c)--(1cm,\c);
                \draw [ densely dotted](2.75,3.25cm)--(7.5cm,3.25cm);
                \draw [ densely dotted](2.75,3.75cm)--(7.5cm,3.75cm);
                \draw [ densely dotted](4.75,\g)--(7.5cm,\g);
                \draw [ densely dotted](4.75,\F)--(7.5cm,\F);
                     \draw [ densely dotted](7cm,8cm)--(7.5cm,8cm);
                \draw [ densely dotted](7cm,8.5cm)--(7.5cm,8.5cm);
		 \draw [gray, densely dotted](10.5cm,11.5cm)--(13.5cm,11.5cm);
		 \draw [ densely dotted](10.5cm,12cm)--(12.5cm,12cm);
                \draw [ densely dotted](\A,\a)--(\A,2cm);
                \draw [ densely dotted](\b,\a)--(\b,2cm);
                \draw [ densely dotted](2.75cm,3.25cm)--(2.75cm,3.75cm);
                \draw [ densely dotted](3.25cm,3.25cm)--(3.25cm,3.75cm);
                \draw [ densely dotted](4.75cm,3.25cm)--(4.75cm,6cm);
                \draw [ densely dotted](5.25cm,3.25cm)--(5.25cm,6cm);
                \draw [ densely dotted](\h,3.25cm)--(\h,8.5cm);
                \draw [ densely dotted](\H,3.25cm)--(\H,8.5cm);
                \draw [ densely dotted](10.5cm,11.5cm)--(10.5cm,12.5cm);
                 \draw [ densely dotted](11cm,11.5cm)--(11cm,12.5cm);
                  \draw [dotted](0.75cm,-0.85cm)--(0.75cm,0.25cm);
                 \draw [dotted](0.75cm,0.65cm)--(0.75cm,1.4cm);
                  \draw [dotted](11.25cm,11.75cm)--(12cm,11.75cm);
                  \draw [dotted](12.5cm,11.75cm)--(13.25cm,11.75cm);
                   \fill[green] (0.25cm,1.75cm) circle (2mm);
                   \fill[black] (0.75cm,0.45cm) circle (1mm);
                   \fill[green] (\d,3.5cm) circle (2mm);
                   \fill[green] (\f,5.75cm) circle (2mm);
                \fill[green] (7.25cm,8.25cm) circle (2mm);
               \fill[black] (\g,4.5cm) circle (1mm);
               \fill[green] (10.75cm,12.25cm) circle (2mm);
               \fill[black] (12.25cm,11.75cm) circle (1mm);
                   \draw (-0.75cm,\b) node {\Large ($iii$)}; \draw (0.75cm,-1.5cm) node {\Large $\sigma_2$};
                   \draw (\d,-1.5cm) node {\Large $\sigma_i$};\draw (\f,-1.5cm) node {\Large $\sigma_j$};
                   \draw (\g,-1.5cm) node {\Large $\sigma_t$};\draw (7.25cm,-1.5cm) node {\Large $\sigma_k$};
                \draw (9.5cm,11.75cm) node {\Large ($ii$)}; \draw (14.5cm,11.75cm) node {\Large $n-1$};
                 \draw [decorate,decoration={brace,amplitude=12pt},xshift=-4pt,yshift=2pt](2.5cm,\e) -- (\G,\I) node [black,midway,xshift=-1cm] {\Large ($iv$)};
                          \end{tikzpicture}
              }
        \end{center}
\caption{An illustration of Theorem~\ref{thm2} that characterizes the basis permutations of $\B_p$.
Condition ($i$) states such a basis permutation has $p+1$ non-left-to-right maxima (drawn with a small black point, while left-to-right maxima are drawn with a big green point),
condition ($ii$) states that $n-1$ is not a left-to-right maximum,
condition ($iii$) states that $\sigma_2$ is not a left-to-right maximum,
and condition ($iv$) states that there is a "higher" non-left-to-right maximum between 3 left-to-right maxima.
}\label{fig3}
\end{figure}

\newpage

\begin{cor}[Length of the forbidden patterns]\label{cor2}
Permutations in $\B_p$ have  length $\leq 2 (p +1)$ and $\geq p+2$.
As a consequence, $\B_p$ is a finite set.
\end{cor}
\begin{proof}Theorem~\ref{thm2} implies that the number of left-to-right maxima in a basis permutation is at most the number of non-left-to-right-maxima.
Since a basis permutation of $\B_p$ has $p+1$ non-left-to-right-maxima, its length is at most $2(p+1)$.
\end{proof}

\bigskip

For instance, the basis for $p=0, 1, 2$ are respectively $\B_0=\{21\}$, $\B_1=\{312,321,2143\}$ (recovering the result of Magn{\'u}sson), and
$\B_2=\{4123,4132,4213,4231,4312,4321,21534, 21543,31254,\linebreak32154,31524,31542, 32514,32541,214365\}$.

\section{Enumerative results for basis permutations}

In order to obtain a recursive formula for the number $b_{n,p}$ of permutations of length $n$ in the basis $\B_p$,
we present the following preliminary lemma.

\begin{lem}[A recursive description] \label{lem3}
Let $\sigma\in \S_n$ be a basis permutation having $p\geq 1$ non-left-to-right-maxima and such that $\sigma_{k+1}=n$, $k\geq 0$. Let $\alpha$ be the subsequence $\sigma_1\sigma_2\ldots \sigma_{k}$ and $\pi$ be the permutation in $\S_{k}$ isomorphic to $\alpha$. Then, $\pi$ is a basis permutation with $p-n+k+1$ non-left-to-right-maxima.
\end{lem}

\begin{proof}Any permutation $\sigma$ can be uniquely written as $\sigma=\alpha n\beta$ where $\alpha$ and $\beta$ are two subsequences of $\interval{n-1}$. Let $k$ be the length of $\alpha$ and let $\pi=\pi_1\pi_2\ldots \pi_k$ be the permutation of $\interval{k}$ that is isomorphic to the subsequence $\alpha$. Let us prove that $\pi$ is minimal.

 Since $\sigma$ is minimal, it satisfies the three conditions ($ii$), ($iii$) and ($iv$) of Theorem~\ref{thm2}.
 Since all elements in $\beta$ are non-left-to-right-maxima in $\sigma$, $\pi$ contains exactly $p-(n-1-k)$ non-left-to-right-maxima and thus, $n-p-1$ left-to-right maxima.

- The condition ($iii$) of Theorem~\ref{thm2} on $\sigma$ does not involve the part $n\beta$. Therefore, $\pi$ satisfies ($iii$).

- The deletion of $n\beta$ from $\sigma$ preserves the condition ($iv$) on $\pi$. Thus, $\pi$ satisfies ($iv$).

- Let $\sigma_i$, $\sigma_j$ and $\sigma_{k+1}=n$, $1\leq i<j\leq k$, be the last three left-to-right maxima of $\sigma$. After the deletion of $n\beta$, the two left-to-right maxima of $\sigma$, $\sigma_i$ and $\sigma_j$, are respectively transported in $\pi$ into $\pi_i$ and $\pi_j=k$. Condition ($iv$) on $\sigma$ ensures that there is $\sigma_\ell$, between $\sigma_j$ and $n$ such that $\sigma_\ell>\sigma_i$. The greatest value $\sigma_\ell$ satisfying this property is then transported in $\pi$ into $k-1$, which proves that $k-1$ is on the right of $k$ in $\pi$. Thus, $\pi$ satisfies ($ii$).

Using Theorem~\ref{thm2}, the permutation $\pi$ is a basis permutation with $p+\ell-n+1$ non-left-to-right-maxima.\end{proof}

\begin{thm}[An infinite recursion]\label{thm3}
The number $b_{n,p}$ of basis permutations of length $n$ in $\B_p$ (or equivalently having exactly $p+1$ non-left-to-right-maxima) is given by the following recurrence relation (for $p<n-2$):
$$b_{n,p}=\sum_{\ell=0}^{p-1}(\ell+1)!\cdot{ \binom{n-2}{\ell}} \cdot b_{n-\ell-2,p-\ell-1}$$
anchored with $b_{n,p}=0$ if $p<(n-2)/2$ or $p > n-2 $, and $b_{n,n-2}=(n-1)!$ for $n>1$.
\end{thm}
\begin{proof}Any permutation $\sigma$ of length $n\geq 1$ contains at least one left-to-right maximum and thus, at most $n-1$ non-left-to-right-maxima which implies that $b_{n,p}=0$ for $n\leq p+1$. Using the proof of
Corollary~\ref{cor2}, we also have $b_{n,p}=0$ for $n> 2(p+1)$. Moreover, the basis permutations of length $n$ with $n-1$ non-left-to-right-maxima are the permutations of the form $n \alpha$ where $\alpha\in \S_{n-1}$. So, we have $b_{n,n-2}=(n-1)!$ for $n>1$.

Now, let us prove the recursive relation. Let $\sigma\in \S_n$ be a basis permutation with $p+1$ non-left-to-right-maxima. We consider its unique decomposition $\sigma=\alpha n \beta$ where $\alpha$ and $\beta$ are some subsequences of $\interval{n-1}$.
Let $\ell+2$, $\ell\geq 0$, be the length of $n\beta$ and let $\pi$ be the permutation in $\S_{n-\ell-2}$ isomorphic to $\alpha$. Using Lemma~\ref{lem3} with $k=n-\ell-2$, $\pi$ is minimal with $p-\ell-1$ non-left-to-right-maxima.
So, we can associate to $\sigma=\alpha n\beta$ the pair $(\pi,\gamma)$ where $\pi\in \S_{n-\ell-2}$ is minimal with $p-\ell-1$ non-left-to-right-maxima and $\gamma\in \S_{\ell+1}$ is isomorphic to $\beta$.

Conversely, let $\pi$ be a basis permutation of length $n-\ell-2$ with $p-\ell-1$ non-left-to-right-maxima and $\gamma\in \S_{\ell+1}$. We construct a basis permutation $\sigma$ of length $n$ with $p+1$ non-left-to-right-maxima as follows.
From $\gamma\in \S_{\ell+1}$, we construct a subsequence $\beta$ of $\interval{n-1}$ of length $\ell+1$ such that $\beta$ contains the value $n-1$ and such that $\beta$ is isomorphic to $\gamma$. Since $n-1$ belongs to $\beta$,
its position in $\beta$ also is the position of the greatest value of $\gamma$. So, $\beta$ is characterized by the choice of $\ell$ values among $\interval{n-2}$. Now, we define the unique subsequence $\alpha$ of $\interval{n-2}\backslash X$ isomorphic to $\pi$ where $X$ is the set of values used in $\beta$.
This construction ensures that $\sigma=\alpha n\beta$ is a basis permutation of length $n$ with $p+1$ non-left-to-right-maxima, and so $\sigma\in\B_p$.
So, there are $\binom{n-2}{\ell}$ possibilities to choose the values of $\beta$ and $(\ell+1)!$ possibilities to choose $\gamma$ and $b_{n-\ell-2,p-\ell-1}$ possibilities to choose a basis permutation $\pi\in \S_{n-\ell-2}$ with $p-\ell-1$ non-left-to-right-maxima.
Varying $\ell$ from 0 to $p-1$, we obtain the recursive formula.
\end{proof}

Theorem~\ref{thm3} allows us to find the bivariate exponential generating function for the number of basis permutations according to the number of non-left-to-right-maxima.

\begin{thm}[Closed-form for the bivariate generating function] \label{thm4}
Consider the bivariate exponential generating function
$B(x,y)=\sum_{n\geq 0,p\geq 0}b_{n,p}\frac{x^ny^p}{n!}$  where the coefficient of $\frac{x^ny^p}{n!}$ is the number $b_{n,p}$ of basis permutations of length $n$ in $\B_p$. Then, we have
$$B(x,y)= \frac{1}{2y}\left(1-\frac{1}{V}\right)  (1-xy)^{\frac{1}{2}(1+V)}
+\frac{1}{2y}\left(1+\frac{1}{V}\right) (1-xy)^{\frac{1}{2}(1-V)}-\frac{1}{y}\,,$$
$$ \text{where } V:=\sqrt{1+4/y}\,.$$
\end{thm}
\begin{proof}
Setting $F_p(x):=\sum_{n\geq 0} b_{n,n-p}\frac{x^n}{n!}$ and  $F(x,y):=\sum_{p\geq 0}F_p(x)y^p$,
we have $B(x,y)=F(x y,1/y)$.
(We work with the generating function $F_p(x)$ of the ($b_{n,n-p}$)'s  rather than the generating function of the ($b_{n,p}$)'s
because then the derivation of the proof is simpler to write).

Taking the second derivative of $F(x,y)$ with respect to $x$ gives
\begin{eqnarray}\label{diff1}\partial_x^2 F(x,y)
=
\partial_x^2 \left(\sum_{p\geq 0}  F_p(x) y^p \right)= \partial_x^2 F_0(x) + \partial_x^2F_1(x) y+ \partial_x^2 F_2(x) y^2+ \sum_{p\geq 3} \partial_x^2 F_p(x) y^p\,.
\end{eqnarray}

Now, the recursive relation of Theorem~\ref{thm3} for $b_{n+2,n-p+2}$ implies for $p\geq 3$:
\begin{eqnarray*}
\partial_x^2  F_p(x) &=&\sum\limits_{n\geq 0} b_{n+2,n-p+2}\frac{x^n}{n!}
=\sum_{n\geq 0}\frac{x^n}{n!}\sum_{\ell=0}^{n-p+1}(\ell+1)!\cdot\binom{n}{\ell}\cdot b_{n-\ell,(n-\ell)-(p-1)}\\
&=& \sum_{n\geq 0}(n+1)!\frac{x^n}{n!}\cdot\sum_{n\geq 0}b_{n,n-p+1}\frac{x^n}{n!}
= \frac{1}{(1-x)^2}F_{p-1}(x).
\end{eqnarray*}
Plugging this recurrence into the 
equation~\eqref{diff1} (and using $F_0(x)=F_1(x)=0$)
gives:
\begin{eqnarray*}\partial_x^2 F(x,y)
&=&\partial_x^2 F_2(x) y^2+ \sum_{p\geq 2}  \frac{y}{(1-x)^2} F_{p}(x) y^{p}.
\end{eqnarray*}
It remains to simplify $F_2(x)$; the initial conditions of Theorem~\ref{thm3} ($b_{n,n-2}=(n-1)!$ and  $b_{n,p}=0$ for $n<2$) imply that
$$F_2(x)=\sum_{n\geq 2} b_{n,n-2}\frac{x^n}{n!}=\sum_{n\geq 2} \frac{(n-1)!}{n!}x^n=-\ln(1-x)-x\,.$$
This leads to the main differential equation:
\begin{eqnarray}\label{diffeqF} \partial_x^2  F(x,y) &=& \partial_x^2   F_{2}(x) y^2 +\frac{y}{(1-x)^2}F(x,y)
= \frac{y}{(1-x)^2}(y+F(x,y))\,.
\end{eqnarray}

First, by plug \& prove, the solutions of $\partial_x^2 F(x)= y F(x)/(1-x)^2$
are a linear combination of $(1-x)^\alpha$, with $\alpha=(1+\sqrt{1+4y})/2$, or $\alpha=(1-\sqrt{1+4y})/2$.
Now, $F(x,y)=-y$ is a trivial particular solution of the non-homogeneous differential equation~\eqref{diffeqF},
so the general solution of this differential equation is of the form:
$$F(x,y)=K(y)\cdot (1-x)^{\frac{1}{2}(1+\sqrt{1+4y})}+L(y)\cdot (1-x)^{\frac{1}{2}(1-\sqrt{1+4y})}-y.$$

Since $F(0,y)=0$ and $\frac{\partial F(x,y)}{\partial x}\rfloor_{x=0}=0$, we respectively deduce the two equations
\[K(y)+L(y)-y=0  \text{\qquad and \qquad} -K(y)\cdot (1+\sqrt{1+4y})-L(y)\cdot (1-\sqrt{1+4y})=0\,,\]
 thus we obtain
\[K(y)=\frac{y}{2}\frac{\sqrt{1+4y}-1}{\sqrt{1+4y}}\qquad \mbox{ and } \qquad L(y)=\frac{y}{2}\frac{\sqrt{1+4y}+1}{\sqrt{1+4y}}\,. \]
This gives
$$F(x,y)=\frac{y}{2}\frac{\sqrt{1+4y}-1}{\sqrt{1+4y}} \cdot (1-x)^{\frac{1}{2}(1+\sqrt{1+4y})}+\frac{y}{2}\frac{\sqrt{1+4y}+1}{\sqrt{1+4y}}\cdot (1-x)^{\frac{1}{2}(1-\sqrt{1+4y})}-y\,,$$
and therefore the theorem, as $B(x,y)= F(xy,1/y)$.\qedhere
\end{proof}

\begin{thm}[Asymptotics] \label{thm5} The exponential generating function for the number of basis permutations with respect to their length is given by
\begin{eqnarray*}
B(x) &= &B(x,1)=\sum_{n\geq 0} b_n \frac{x^n}{n!} = \frac{\sqrt{5}-1}{2\sqrt{5}}\cdot(1-x)^{\frac{1+\sqrt{5}}{2}}+\frac{\sqrt{5}+1}{2\sqrt{5}}\cdot(1-x)^{\frac{1-\sqrt{5}}{2}}-1\\
 &=&  \frac{x^2}{2!}+  2 \frac{x^3}{3!}+ 7 \frac{x^4}{4!}+  32\frac{x^5}{5!} +  179 \frac{x^6}{6!}+ 1182 \frac{x^7}{7!} +  8993 \frac{x^8}{8!} +O(x^9).
\end{eqnarray*}
It is a D-finite transcendental function satisfying the following differential equation
\begin{equation}\label{diffeq1} B(x) - (1-x)^2 \partial_x^2 B(x) +1 =0, \qquad B(0)= B'(0)=0\,.\end{equation}
Equivalently, its coefficients $b_n$ satisfy the recurrence
\begin{equation} \label{nicerec} b_{n+2} = 2 n  b_{n+1}+(1+n-n^2) b_n, \qquad b_0=b_1=0, b_2=1\,,\end{equation}
and the asymptotics are given by
$$ \frac{b_n}{n!} \sim \frac{\phi}{\sqrt 5 \, \Gamma(\phi-1)} \frac{1}{n^{2-\phi} }(1+o(1))\,, $$
where $\phi$ is the golden ratio $\phi=(1+\sqrt 5)/2$, and $\Gamma(z):=\int_0^{+\infty} t^{z-1} \exp(-t) dt$ is the Euler gamma function.

Accordingly, a permutation of length $n$ has a probability asymptotically 0 to be an element of the basis of forbidden patterns,
however, this probability is not ``very small'' as it decays only polynomially:
$$\operatorname{Prob}( s\in \S_n \text{ belongs to } \cup_{p\in {\mathbb N}} \B_p) =  \frac{b_n}{n!} \approx 0.499/n^{0.381} (1+o(1))\,.$$
\end{thm}
\begin{proof}
Setting $y=1$ in the bivariate exponential generating function given in Theorem~\ref{thm4} gives $B(x)$.
If a function $B(z)=\sum b_n z^n$ is D-finite (it is satisfying a linear differential equation, with polynomial coefficients in $z$),
then its coefficients $b_n$ are polynomially recursive (in short, P-recursive): they satisfy a linear recurrence, with polynomial coefficients in $n$.
See e.g.~\cite{FlSe, S} for more on these two equivalent notions. Starting from the building blocks $(1-x)^a$, which are D-finite, and then using the closure properties of D-finite functions (by sum and product)
gives the differential equation~\eqref{diffeq1} (this is e.g.~implemented in the Gfun Maple package, see~\cite{SalvyZimmermann1994}).
The recurrence is obtained by extracting the coefficient of $x^n$ on both sides of the differential equation.
The asymptotics follows from a singularity analysis (see~\cite{FlSe}) on each term of the shape $(1-x)^a$, indeed, for any $a\in \mathbb{R}$ which is not an integer, one has:
\[ [x^n] (1-x)^a =    \frac{1}{\Gamma(-a) n^{1+a}}  \left(1+ \frac{1}{2} a (a+1)\,  \frac{1}{n} + O(\frac{1}{n^2})\right)\,. \qedhere \]
\end{proof}

Note that $B(x,y)$ is D-finite in the variable $x$:
\[ 1+ y B(x,y) - (1-xy)^2 \partial_x^2 B(x,y) =0
\text{\qquad with  $B(0,y)= (\partial_xB)(0,y)=0$,} \]
 but it is not D-finite in the variable $y$.
This follows from a saddle point analysis on $B(1,y)=\sum_n \beta_n y^n$, indeed the asymptotics of $\beta_n$ involve arbitrarily large $(\ln n)^d$,
while the asymptotics of a D-finite function can only have a finite sum of such powers of log, see~\cite{FlSe}.
This argument is thus similar to a proof that $(1-y)^{1-y}$ is not D-finite.

 \bigskip

{\bf Remark} [Irrational critical exponent]: It is very seldom that a combinatorial problem leads to some asympotics involving an irrational number as exponent.
In fact, in combinatorics and in statistical physics, most of the asymptotics of integer sequences are of the shape $b_n \sim C n^\alpha A^n$,
and the exponent $\alpha$ which appears there is a key quantity: its value is often the signature of some universal phenomena (in physics, it is called a critical exponent).
For D-finite sequences, the theory implies that it is an algebraic number, 
however, this exponent is very often -3/2, or a dyadic number (for the reasons explained in~\cite{BanderierDrmota}), or a rational number (due to a result on G-functions).
Indeed, a theorem (resulting from the works of Katz, Andr\'e, Chudnovsky \& Chudnovsky, see~\cite{ChambertLoir}) states
that G-functions (D-finite functions with integer coefficients and non-zero radius of convergence) have a rational critical exponent.
Now, instead of considering the exponential generating function $B(x)=\sum b_n x^n/n!$,
we may consider its inverse Borel transform, {\em i.e.}, the ordinary generating function $\sum b_n x^n$.
It is also a D-finite function, because D-finite functions are closed by Hadamard product,
and therefore the Borel transform (and the inverse Borel transform) of a D-finite function is D-finite
({\em i.e.}, if the sequence $b_n$ is P-recursive, so are $n! b_n$ and $b_n/n!$).
We have thus a new D-finite function with integer coefficients and irrational critical exponent (involving the golden ratio $\phi$),
but this is not contradicting the G-function theorem, because, due to the multiplication by $n!$, we now have a 0 radius of convergence.
In conclusion, we have the pleasure to have here one of the few examples in combinatorics of a problem leading to an irrational critical exponent.
Other examples are given via the KPZ formula in physics, or via quantities related to quadtrees, see~\cite{FlSe}.
\bigskip

Pushing further the asymptotics from Theorem~\ref{thm5}, we get the following limit law:

\begin{thm}[Limit Law]\label{Th6}
In the model where all permutations of length $n$ are equidistributed, a random permutation of length $n$ in $\cup_{p\in {\mathbb N}} \B_p$ is typically a member of $\B_p$,  for $p\sim n-(\ln n) /\sqrt{5}$, with Gaussian fluctuations.
Equivalently, the average number of left-to-right-maxima in a random basis permutation is  $p\sim (\ln n) /\sqrt{5}$ with Gaussian fluctuations.
\end{thm}
\begin{proof}
This follows from the closed-expression for $B(x,y)$, or from a singularity analysis of the differential equation.
Indeed, the average and standard deviation follow from the computation of $\partial_y B(x,y)$ and $\partial_y^2 B(x,y)$ 
at $y=1$.
The Gaussian limit law follows from the quasi-power theorem applied to a variable exponent perturbation or to our non-confluent differential equation (see Theorem IX.11 and Theorem IX.18 from~\cite{FlSe}).
\end{proof}

As a random permutation of $\S_n$ has $\ln n$ left-to-right maxima on average,
the above theorem quantifies to what extent the right-jump process kills the left-to-right maxima
when one starts from the identity permutation.

\newpage

\begin{figure}[ht]
\begin{center}\includegraphics[height=55mm]{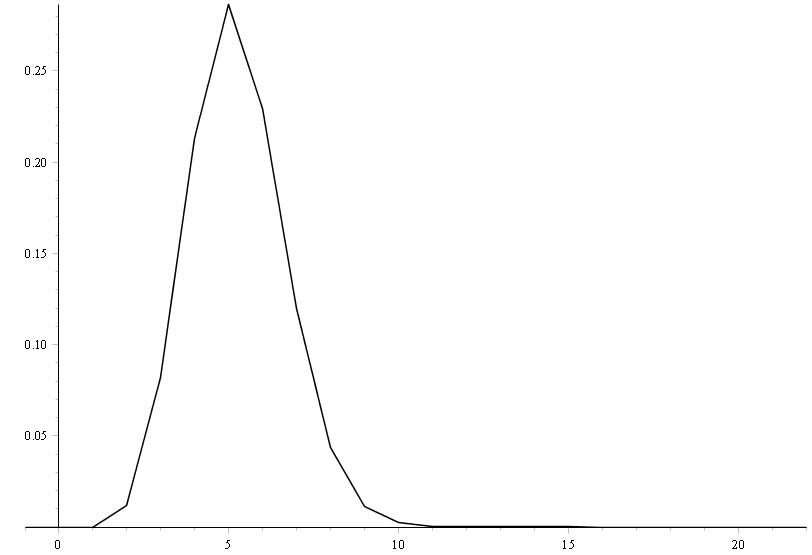}\vspace{-3mm}
\end{center}
\caption{This histogram illustrates Theorem~\ref{Th6}:
the average number of left-to-right-maxima in a random basis permutation is  $p\sim (\ln n) /\sqrt{5}$ with Gaussian fluctuations. 
However the speed of convergence and the small amplitude of the variance makes that large values of $n$ have to be considered 
to recognize clearly the bell curve of the Gaussian (above, this is a histogram for $n=4000$).}
\end{figure}

\smallskip

For the combinatorial structure ${\mathcal B}_p$, it could be possible that its complementary set has a nicer structure.
Those permutations not in the basis are for sure counted by $u_n=n!-b_n$;
this sequence satisfies
$u_{n+3}=(n+1)(n^2-n-1)u_n-(3n^2+3n-1) u_{n+1}+3(n+1) u_{n+2}$,
which is still a nice recurrence but of order one more than the recurrence for $b_n$,
so it is a heuristic confirmation than $b_n$ is a more fundamental sequence than $u_n$.

The first values of $b_n$  (the number of basis permutations of length $n$) are $1, 2, 7, 32$, $179$,  $1182$, $8993$, $77440$ for $2\leq n\leq 9$.
We added this sequence to the On-line Encyclopedia of Integer Sequences (hereafter abbreviated OEIS), see \cite{Slo}:

\begin{table}[hbt]
\begin{center}\scalebox{1.00}{
\begin{tabular}{c|cccccccccc|c}
$p\backslash n$ & 2 &  3 & 4 & 5 & 6 & 7 & 8 & 9 & 10 & 11  & $\# \B_p$ \\
\hline
0 & 1 & & & & & & & & & & 1 \\
\cline{2-12}
1 & & 2 & 1 & & & & & & & & 3\\
\cline{2-12}
2  && & 6 & 8 & 1 & & & & &  & 15 \\
\cline{2-12}
3  && & & 24 & 58 & 18 & 1 & &&  & 101 \\
\cline{2-12}
4  & && & & 120 & 444 & 244 & 32 & 1 &  & 841 \\
\cline{2-12}
5  & && & & & 720 & 3708 & 3104 & 700 & 50 &  8232 \\
\cline{2-12}
6  & & && & & & 5040 & 33984 & 39708 & 13400 & 78732 \\
\hline
$\Sigma$ & 1 & 2 & 7 & 32 & 179 & 1182 & 8993 & 77440 &744425& 7901410  &\\
\end{tabular}}
\end{center}
\caption{\label{tab1}Number $b_{n,p}$ of basis permutations of length $n$ (the "minimal forbidden patterns" of $\B_p$,
or equivalently, with $p+1$ non-left-to-right-maxima) where $2\leq n\leq 11$ and $0\leq p\leq 6$ \OEIS{A265163}.
The last column contains $\beta_p:=\sum_n b_{n,p}$ \OEIS{A265164};
the last line contains $b_n:=\sum_p b_{n,p}$ \OEIS{A265165}.}
\end{table}

There is a vast literature in number theory analysing the modular congruences of famous sequences (Pascal triangle, Fibonacci, Catalan, Motzkin, Ap\'ery numbers, see~\cite{Sagan, Xin, RowlandZeilberger, Kauers}).
The properties of $b_n \mod m$ are sometimes called "supercongruences" when $m$ is the power of a prime number: many
articles consider $m=2^r$, or $m=3^r$.
We now give a result which holds for any $m$ (not necessarily the power of a prime number).

\begin{thm}[Supercongruences for D-finite functions]\qquad\newline   \label{supercongru}
Consider any P-recurrence of order $r$: $$P_0(n) u_n = \sum_{i=1}^r  P_i(n) u_{n-i}\,.$$

If the polynomial $P_0(n)$ is ultimately invertible $\mod m$ ({\em i.e.,} $\gcd(P_0(n),m) =1 $, for all $n$ large enough),
then the sequence $(u_n)$  is ultimately periodic\footnote{In the sequel, we will omit the word "ultimately": a periodic sequence of period $p$ is thus a sequence for which  $u_{n+p} = u_n$ for all large enough $n$. Some authors use the terminology "eventually periodic" instead.}
 $\mod m$, and there is an algorithm to get this period.

In particular, recurrences such that $P_0(n)=1$ are periodic $\mod m$.

Accordingly, our sequence $b_n \mod m$ (defined by recurrence~\eqref{nicerec}) is periodic for any $m$.
\end{thm}

\begin{proof}
Indeed, as the leading term $P_0$ is invertible, we can write:
$$u_n \mod m = \sum_{i=1}^r  \frac{P_i(n) \mod m}{P_0(n) \mod m}  (u_{n-i} \operatorname{mod} m)\,,$$
in which each term has just a finite set of possible values.
What is more, for any polynomial  $P(n)$ with integer coefficients,  $P(n) \mod m$ is of period $p$, for some $p|m$.
(This follows from the fact that the sum and the product is preserving periodicity $\mod m$,
as we did not require in the definition of "period" that $m$ is the smallest $m$ such that the sequence is $m$ periodic).
Therefore, one can then construct a Markov chain (an automaton):
the states are all the possible $2r+1$-tuples of values $\mod m$  for 
$(P_0(n),\dots, P_r(n), u_{n-1},\dots,u_{n-r})$, and the recurrence dictates the transitions in this Markov chain.
The pigeonhole principle implies that there is  a loop in this finite graph, and our period is thus one of the divisors of the length of this loop.

Besides, the smallest period  $p$ and the first integer $n_1$ satisfying $u_n=u_{n+p} \mod m$ for $n\geq n_1$ 
are such that  $p+n_1$ is smaller than the number of states in the automaton,
i.e.~smaller than $m^{2r+1}$, or more precisely smaller than 
$\phi(m) m^{2r}$ where $\phi(m)$ is the Euler totient function (the number of invertible elements modulo~$m$, as $P_0(n)$ is ultimately invertible).
\end{proof}


This theorem explains the periodic behaviour of $b_n \mod m$.
By brute-force computation, we can get $b_n \mod m$, for any given $m$.
For example $b_n\mod 15$ is periodic of period $12$:
for $n\geq 9$, one has $$b_n\mod 15= (10, 5, 10, 10, 0, 10, 5, 10, 5, 5, 0, 5)^\infty.$$
The period can be quite large, for example $b_n\mod 3617$ has period $26158144$.
We say more on this phenomenon for P-recursive sequences in~\cite{BaLu17}.


\newpage
\section{Conclusion}

In this article, we analysed the iteration of the process "a particle jumps to the right" in a permutation,
and we gave the typical properties of the patterns which are not reached after $p$ moves.
We expect our approach (introducing a catalytic variable and getting a D-finite function) to work in many other cases.
However, we already know a nice permutation class for which the basis is  not D-finite.
Indeed, as an extension of this work,
an interesting question is to consider a model in which both right-jumps and left-jumps are allowed:
this is a very natural process, also related to sorting algorithms and bioinformatics processes.
In a forthcoming work, we show that for this new model,
the basis of forbidden minimal patterns for permutations obtained by $p$ iterations of the process
 is related to Young Tableaux with 2 equally long first rows
(but it is no longer D-finite, unlike the pure right-jump iteration process that we considered in this article).

Another natural question is: is it the case that using e.g.~the 
correspondence between records and cycles in permutations,
there is an elegant process corresponding to a "particle jumps to the right", with permutations at distance $p$ from identity being counted in terms of cycles in the permutation?
To get direct "bijective" proofs of our formulae is also an interesting question:
as a credo, it cannot be the case that such nice formulae/recurrences are only reached by solving differential equations (like we did in this article).
It may be the case that a generating tree approach leads to the simple recurrence~\eqref{nicerec} we get for $b_n$ (see e.g.~\cite{banderier} or~\cite{ECO}).
With respect to the asymptotics, it is noteworthy that the process analysed in the present article 
involves the golden ratio (this is very unusual in combinatorics to have this constant {\em as critical exponent}): is it the trace of some universality class? 
(Like it is sometimes the case for problems coming from statistical mechanics, see our remark after Theorem ~\ref{thm5}.)

Last but not least, we already mentioned a vast literature of publications in number theory
analysing the modular congruences of famous sequences (Pascal triangle, Fibonacci, Catalan, Motzkin, Ap\'ery numbers,~\dots,~\cite{Sagan, Xin}).
It seems to us that our approach to tackle them at the level of D-finite functions
is new (see also~\cite{Kauers, RowlandZeilberger}), and it would be worth 
analysing these properties in full generality.
In this article, we proved by a mixture of Ansatz and brute force proof that 
$b_n \mod m$ (where $m$ can be any integer) is a periodic function
(of period bounded by a polynomial in $m$). In fact, we prove in~\cite{BaLu17}  that 
this period and the values of $b_n \mod m$ for any given $m$ can be made explicit.


\bigskip

{\bf Acknowledgments:} We thank our laboratories (LE2I and LIPN) for funding this collaboration,
and  presentation of preliminary results at the conferences AofA'15 and Permutation Patterns'15.
We also thank the two referees for their careful reading.

\smallskip

\clearpage
\bibliographystyle{plainnat}
\bibliography{rightjump}
\bigskip 
\end{document}